\documentclass[lettersize,journal]{IEEEtran}
\usepackage{mathptmx} 
\usepackage{footnote}
\makesavenoteenv{tabular}
\makesavenoteenv{table}
\usepackage{pgfplots}

\usepackage{bbm}
\usepackage{booktabs}
\usepackage{xcolor}
\usepackage[hidelinks, hyperfootnotes=true]{hyperref}
\usepackage{nomencl}
\usepgfplotslibrary{fillbetween}
\definecolor{lightgreen}{rgb}{0.56, 0.93, 0.56} 
\usepackage{hyperref}
\usepackage{amsmath,amsfonts}
\usepackage{amsthm}
\newtheorem{theorem}{Theorem}
\newtheorem{proposition}{Proposition}

\usepackage{algorithm}
\usepackage{enumitem}
\usepackage{algpseudocode}
\usepackage{threeparttable}
\usepackage{array}
\usepackage[caption=false,font=footnotesize]{subfig} 
\usepackage{textcomp}
\usepackage{stfloats}
\usepackage{url}
\usepackage{verbatim}
\usepackage{graphicx}
\usepackage{cite}
\usepackage{color, colortbl} 
\usepackage{etoolbox}
\usepackage{nomencl}
\makenomenclature
\renewcommand\nomgroup[1]{%
  \item[\itshape
  \ifstrequal{#1}{R}{Residential Load Profiles:}{%
  \ifstrequal{#1}{F}{Full Convolutional Profile Flow:}{%
  \ifstrequal{#1}{E}{Evaluation Metrics:}{}}}%
]}

\begin{document}


\title{Estimating Density Functions for Probabilistic Power Flow Using Invertible Neural Networks}

\author{Weijie Xia,~\IEEEmembership{Student Member,~IEEE}, James Ciyu Qin,~\IEEEmembership{Member,~IEEE},  Edgar Mauricio Salazar Duque,~\IEEEmembership{Student Member,~IEEE}, Hongjin, Du,~\IEEEmembership{Student Member,~IEEE}, Peter Palensky,~\IEEEmembership{Senior Member,~IEEE}, Giovanni Sansavini,~\IEEEmembership{Member,~IEEE}, Pedro P. Vergara,~\IEEEmembership{Senior Member,~IEEE}
\thanks{This research was supported by the Align4Energy Project (NWA.1389.20.251) and utilized the Dutch National e-INNrastructure
with the support from the SURF Cooperative (grant number: EINN-5398).}
\thanks{Weijie Xia, Hongjin Du, Peter Palensky, and Pedro P. Vergara are with the Intelligent Electrical Power Grids (IEPG) Group, Delft University of Technology, 2628 CD Delft, The Netherlands (e-mail:\{W.xia,  H.Du, P.P.VergaraBarrios, P.Palensky\}@tudelft.nl).}
\thanks{G. Sansavini and James Ciyu Qin are with the Department of Mechanical and Process Engineering, ETH Zurich, Zurich 8092, Switzerland (e-mail: \{ciyqin, sansavig@
ethz.ch\}).}
\thanks{Edgar Mauricio Salazar Duque is with the Electrical Energy Systems Group, Eindhoven University of Technology, 5612 AE Eindhoven, The Netherlands (e-mail: e.m.salazar.duque@tue.nl).}}

\markboth{IEEE TRANSACTIONS ON POWER SYSTEMS, VOL. -, NO. -, -}%
{\MakeLowercase{\textit{Xia et al.}}: Voltage Density Approximation for Probabilistic Power Flow Using Invertible Neural Networks}


\maketitle

\begin{abstract}


Probabilistic power flow (PPF) is essential for quantifying operational uncertainty in modern power systems with high penetrations of renewable generation and flexible loads. Conventional PPF methods primarily rely on Monte Carlo (MC)-based power flow (PF) simulations or simplified approximations of voltage probability density functions. Although MC methods provide high accuracy, they incur substantial computational and data-storage costs, whereas simplified approximations, though lightweight, often sacrifice accuracy due to the reliance of certain assumption. In this paper, we propose a novel PPF density approximation framework that leverages the expressive power of neural networks (NNs) and can, in principle, approximate complex voltage distributions directly without restrictive assumptions. The core idea is to learn an explicit mapping between stochastic power injections and system voltages using invertible neural networks (INNs), and then to apply the change-of-variables theorem to obtain the voltage densities. Extensive numerical studies demonstrate that the proposed framework achieves state-of-the-art performance both as an accurate PF solver and as an efficient PPF density estimator.

\end{abstract}
\vspace{-0.2cm}
\section{Introduction}

\IEEEPARstart{T}{he} energy transition is transforming conventional power systems into increasingly dynamic and complex networks, driven by the widespread adoption of low-carbon technologies such as electric vehicles, heat pumps, and photovoltaic systems. Their growing integration introduces substantial uncertainty and variability in both electricity generation and demand~\cite{liang2012wide}. Probabilistic power flow (PPF) is therefore widely used to quantify the resulting operational uncertainty by characterizing the probability distributions of system states, thereby supporting reliable planning and operation~\cite{dall2013distributed}.

Conventional PPF commonly relies on Monte Carlo (MC) simulation, in which the PF equations are repeatedly solved for a large number of sampled operating scenarios. Although MC simulation is conceptually straightforward and can provide accurate estimates given sufficiently many samples, its computational cost increases substantially with the number of scenarios and the complexity of the system, limiting its applicability to large-scale or time-critical studies~\cite{yang2019fast}. Processing and storing the resulting PF solutions also introduce considerable data-management and memory overhead.

To alleviate these limitations, existing studies have pursued three main directions: simplifying the PF equations using linearized or other approximate formulations~\cite{gao2023analytical,wang2016analytical,hong1998efficient}, reducing the number of PF evaluations through representative scenario selection~\cite{wang2020scenario,krishna2022uniform}, and accelerating individual PF evaluations using surrogate models such as neural networks (NNs)~\cite{yang2019fast,lin2024powerflownet,wu2022graph,hu2020physics}. NN-based surrogates can achieve substantial computational speedups while maintaining high predictive accuracy. However, their application to PPF generally remains sampling-based, and a large set of injection scenarios must still be generated and propagated through the surrogate to estimate the resulting voltage distributions. Conversely, closed-form density approximations based on simplified PF formulations can directly approximate selected statistical quantities or the probability density functions (PDFs) of voltage states, but their accuracy may deteriorate when their underlying assumptions do not adequately represent the nonlinear PF mapping. For example, the method in~\cite{wang2016analytical} derives voltage PDFs under a linear PF approximation, which can introduce appreciable errors under nonlinear operating conditions.

Despite the widespread use of NN-based models as PF surrogates solver, their application to direct density approximation in PPF remains largely unexplored. Owing to their strong representational capacity, NNs offer the potential to improve both the accuracy and flexibility of voltage-density estimation without relying on restrictive approximations of the nonlinear PF mapping. Motivated by this observation, this paper proposes, to the best of our knowledge, the first NN-based framework for direct approximation of voltage probability density functions in PPF. In particular, we employ a specialized class of NNs known as invertible neural networks (INNs). 



The proposed framework comprises four main components: (i) an INN-based model, termed the invertible mixed neural flow (IMNF); (ii) mathematical formulations for voltage-density approximation using IMNF; (iii) dedicated training strategies for IMNF; and (iv) an efficient scenario-sampling procedure. Once trained, the framework directly approximates voltage probability density functions arising from the nonlinear PF equations, without solving them repeatedly. The main contributions of this paper are summarized as follows:



\begin{itemize}
    \item We propose an INN-based framework for direct voltage density approximation in PPF. By combining an invertible PF surrogate solver with the change-of-variables theorem and scenario-based marginalization, the framework can approximate complex and non-Gaussian voltage probability density functions without repeated PF computations during inference. We systematically present its mathematical foundation, density-modeling architecture, training strategy, and scenario-sampling procedure.

    \item Within this framework, we develop the IMNF model, which, to the best of our knowledge, is the first INN-based surrogate solver designed for PF. IMNF exploits its architectural invertibility to perform both the forward and inverse mappings between power injections and voltages using a single model. Numerical results demonstrate that, beyond the PPF task, IMNF is also a competitive surrogate solver for PF.
    
\end{itemize}

\vspace{-0.2cm}
\section{Related Work}

\subsubsection{Voltage Density Approximation in PPF}

Approximating the probability distributions of system states constitutes an important branch of PPF research. MC-based PPF produces empirical distributions from repeated PF evaluations, whereas density approximation methods seek an explicit representation of the resulting voltage distributions. Such representations can be more readily integrated into risk assessment and other power-system planning and operation tasks~\cite{yuan2020improved}. In~\cite{wang2016analytical}, power injections are modeled using a Gaussian mixture model (GMM), and the PF mapping is approximated by a linear function. Exploiting this linearity enables the corresponding voltage density to be derived in closed form. However, the linear PF assumption may be overly restrictive and can reduce the accuracy of the resulting probability density function (PDF) under nonlinear operating conditions. To address this limitation, nonlinear and locally linear approximation methods have been investigated. For example,~\cite{gao2023analytical} proposes a piecewise-linear approach that partitions the power-injection space into multiple regimes and applies a local linear mapping within each regime. In~\cite{wang2020scenario}, representative scenarios are used to approximate the marginal voltage distribution at a target bus. With the increasing adoption of machine learning (ML) and statistical learning, data-driven and nonparametric distribution-estimation techniques have also been applied to PPF~\cite{abbasi2022comparison}. For instance, Gaussian-process-based methods have demonstrated competitive performance in PPF applications~\cite{pareek2020gaussian,xu2020probabilistic}. Despite these advances, existing voltage-density approximation methods generally rely on restrictive assumptions about the PF mapping, the underlying distributions, or both. An NN-based framework capable of learning the nonlinear PF transformation and directly evaluating voltage densities could reduce these restrictions and improve approximation accuracy.

\subsubsection{NN-based PF Solver}
The primary motivation for employing NN-based PF solvers is to exploit the inherent parallelism of NNs to accelerate PF computation. For example, in~\cite{yang2019fast}, an NN is used to accelerate PF computation, replacing the conventional Newton–Raphson method. Subsequent studies have focused on improving model accuracy mainly by incorporating additional physical knowledge. For instance, the work in~\cite{hu2020physics} leverages the invertibility of the PF mapping and trains two NNs simultaneously to enhance performance. In~\cite{wu2022graph}, a graph attention network (GAT) is integrated to explicitly exploit the topological structure of the power system. Similarly, PowerFlowNet~\cite{lin2024powerflownet} employs graph NNs (GNNs) to model system topology and further improve PF estimation accuracy. In addition, the work in~\cite{xiao2023novel} introduces a recurrent NN-based model that incorporates time-series renewable generation dynamics into PF computation. Although these NN-based models have been proposed as effective PF solvers, they primarily focus on deterministic PF computation. For PPF, NN-based models still rely on MC-based sampling to propagate input uncertainties through the learned mapping, thereby inheriting the computational burden and storage requirements of MC-based PPF. Consequently, an NN-based framework capable of directly approximating voltage probability densities remains highly desirable.

\subsubsection{Invertible NNs}
INNs are first proposed in~\cite{dinh2014nice} as a class of NN-based generative models, also known as flow-based generative models, under the non-linear independent components estimation (NICE) framework. Compared to conventional NNs, INNs are structurally invertible and allow the exact computation of the determinant of the Jacobian. One major challenge of INNs is that their invertibility constraints often limit model flexibility and expressiveness compared with other types of NN-based models. To address these limitations, a large body of subsequent research has been developed based on the original NICE architecture. For example, RealNVP~\cite{dinh2016density} extends NICE by introducing a different formulation of affine coupling layers. However, both NICE and RealNVP are still based on affine transformations, which are not sufficiently flexible for complex generative tasks. To overcome this issue, spline-based flows were introduced in~\cite{durkan2019neural,dolatabadi2020invertible}, significantly improving the expressiveness of INNs beyond affine mappings. INNs have also been applied to energy consumption profile generation, as demonstrated in~\cite{xia2025flow,ge2020modeling}. More recently, attention mechanisms have been integrated into INNs in~\cite{zhai2024normalizing}, achieving state-of-the-art performance compared with other generative models. In this paper, we leverage the invertibility of INNs and propose a hybrid INN-based model to improve their expressiveness for voltage-density approximation in PPF.

\section{Problem Formulation}\label{ProblemFormulation}
We consider a power system with $N$ buses. For each bus $i \in \{1,\dots, N\}$, the net active and reactive power injections are denoted by $p^{i}$ and $q^{i}$, respectively, i.e., $p^{i} = p^{i}_{\mathrm{gen}} - p^{i}_{\mathrm{load}}$ and $q^{i} = q^{i}_{\mathrm{gen}} - q^{i}_{\mathrm{load}}$, where the subscripts $\mathrm{gen}$ and $\mathrm{load}$ indicate the power generated and consumed at bus $i$. Let
\begin{equation}
\mathbf{p} = [p^{1}, p^{2}, \dots, p^{N}], 
\qquad
\mathbf{q} = [q^{1}, q^{2}, \dots, q^{N}]
\end{equation}
denote the vectors of active and reactive injections. Since each injection aggregates local generation and consumption, the joint injection density of the system $\mathcal{P}_{\mathbf{w}}([\mathbf{p}, \mathbf{q}])$ captures the uncertainty of both renewable generation and load demand, as well as their correlations across buses. We assume that $\mathcal{P}_{\mathbf{w}}([\mathbf{p}, \mathbf{q}])$ and the bus-level densities $\mathcal{P}_{\mathbf{w}_i}([p^i, q^i])$ are known and modeled using GMMs~\cite{wang2020scenario, nijhuis2016gaussian}. Therefore, the joint probability densities can then be expressed as
\begin{equation}\label{equ:jointdistribution}
\mathcal{P}_{\mathbf{w}}([\mathbf{p}, \mathbf{q}]) 
= \sum_{k=1}^{K} \pi_k  \, 
\mathcal{N}\!\left( 
\begin{bmatrix} \mathbf{p} \\ \mathbf{q} \end{bmatrix};
\, \mu_k, \Sigma_k
\right),
\end{equation}
\begin{equation}\label{equ:busdistribution}
\mathcal{P}_{\mathbf{w}_i}([p^i, q^i]) 
= \sum_{k=1}^{K} \pi_{k} \,
\mathcal{N}\!\left(
\begin{bmatrix} p^i \\ q^i \end{bmatrix};
\, \mu_{k}^{i}, \Sigma_{k}^{i}
\right),
\end{equation}
where $\pi_k$, $\mu_k$, and $\Sigma_k$ denote the mixture weight, mean vector, and covariance matrix of the $k$-th Gaussian component, and $\mu_k^{i}$ and $\Sigma_k^{i}$ are the sub-vector and sub-matrix of $\mu_k$ and $\Sigma_k$ corresponding to $[p^i, q^i]$. In general, the PF formulation establishes a mapping between the active/reactive power $[\mathbf{p}, \mathbf{q}]$ and the bus voltage magnitudes and phase angles $[|\mathbf{v}|, \boldsymbol{\theta}]$. This relationship can be written as
\begin{equation}
    (\mathbf{p}, \mathbf{q}) \;\longrightarrow\; (|\mathbf{v}|, \boldsymbol{\theta}),
\label{eq:PF_mapping}
\end{equation}
where
$|\mathbf{v}| = \left[ |v^{1}|, |v^{2}|, \dots, |v^{N}| \right]$ and  $\boldsymbol{\theta} = \left[ \theta^{1}, \theta^{2}, \dots, \theta^{N} \right]$. In this paper, we restrict the domain to operationally feasible solutions under normal steady-state operating conditions. Over this domain, the PF mapping in Eq.~\eqref{eq:PF_mapping} admits a unique solution and can also be regarded as almost bijective onto its image~\cite{dvijotham2015solving,duque2024tensor}.
Our objective is to obtain an explicit approximation of the joint probability density of the resulting voltage magnitude and phase angle at each bus,
\begin{equation}
\mathcal{P}_{\mathbf{o}_i}([|v^{i}|, \theta^{i}]), 
\qquad i \in \{1,\dots,N\}
\label{equ:jointvoltage}
\end{equation}

\section{Methodology}
\subsection{Power Flow (PF) Formulation}
The PF formulation describes the nonlinear steady-state relationship between active/reactive power injections and bus voltages in an $N$-bus network. For bus $i$, the complex power injection is
\begin{equation}
S^i = p^{i} + j q^{i} = V^i \sum_{j=1}^{N} V^{j*} Y^{ij*},
\end{equation}
where $Y^{ij}$ is the $(i,j)$-th element of the network admittance matrix, and the complex bus voltage is expressed as
\begin{equation}
V^i = |v^{i}| e^{j \theta^{i}},
\end{equation}
with $|v^{i}|$ and $\theta^{i}$ denoting the voltage magnitude and phase angle, respectively.

\subsection{{Adjustment of the Change of Variable Theorem for PPF}}\label{sec:math}

The \textit{Change of Variable Theorem} provides the fundamental principle for transforming probability densities through an invertible and differentiable mapping, stated formally as follows.
\begin{theorem}[Change of Variable Theorem~\cite{dinh2014nice}]\label{thm:cov}
Let $\mathbf{z} \in \mathbb{R}^{d}$ be a random vector with known density $\mathcal{P}_{\mathbf{z}}(\mathbf{z})$, and let $f : \mathbb{R}^d \rightarrow \mathbb{R}^d$ be a bijective mapping such that $f$ and its inverse $f^{-1}$ are continuously differentiable and the Jacobian determinant $\det(\partial f(\mathbf{z}) / \partial \mathbf{z})$ is non-zero almost everywhere. Then the density of the transformed variable $\mathbf{x} = f(\mathbf{z})$ is given by
\begin{subequations}\label{eq:cov_thm}
\begin{equation}
\mathcal{P}_{\mathbf{x}}(\mathbf{x})
=
\mathcal{P}_{\mathbf{z}}(f^{-1}(\mathbf{x}))
\left|
\det\!\left(
\frac{\partial f^{-1}(\mathbf{x})}{\partial \mathbf{x}}
\right)
\right|,
\label{eq:cov_general}
\end{equation}
or equivalently,
\begin{equation}
\mathcal{P}_{\mathbf{x}}(\mathbf{x})
=
\mathcal{P}_{\mathbf{z}}(\mathbf{z})
\left|
\det\!\left(
\frac{\partial f(\mathbf{z})}{\partial \mathbf{z}}
\right)
\right|^{-1}.
\label{eq:cov_general_z}
\end{equation}
\end{subequations}
\end{theorem}

\begin{figure*}[htp]
    \centering
    \includegraphics[width=0.95\linewidth]{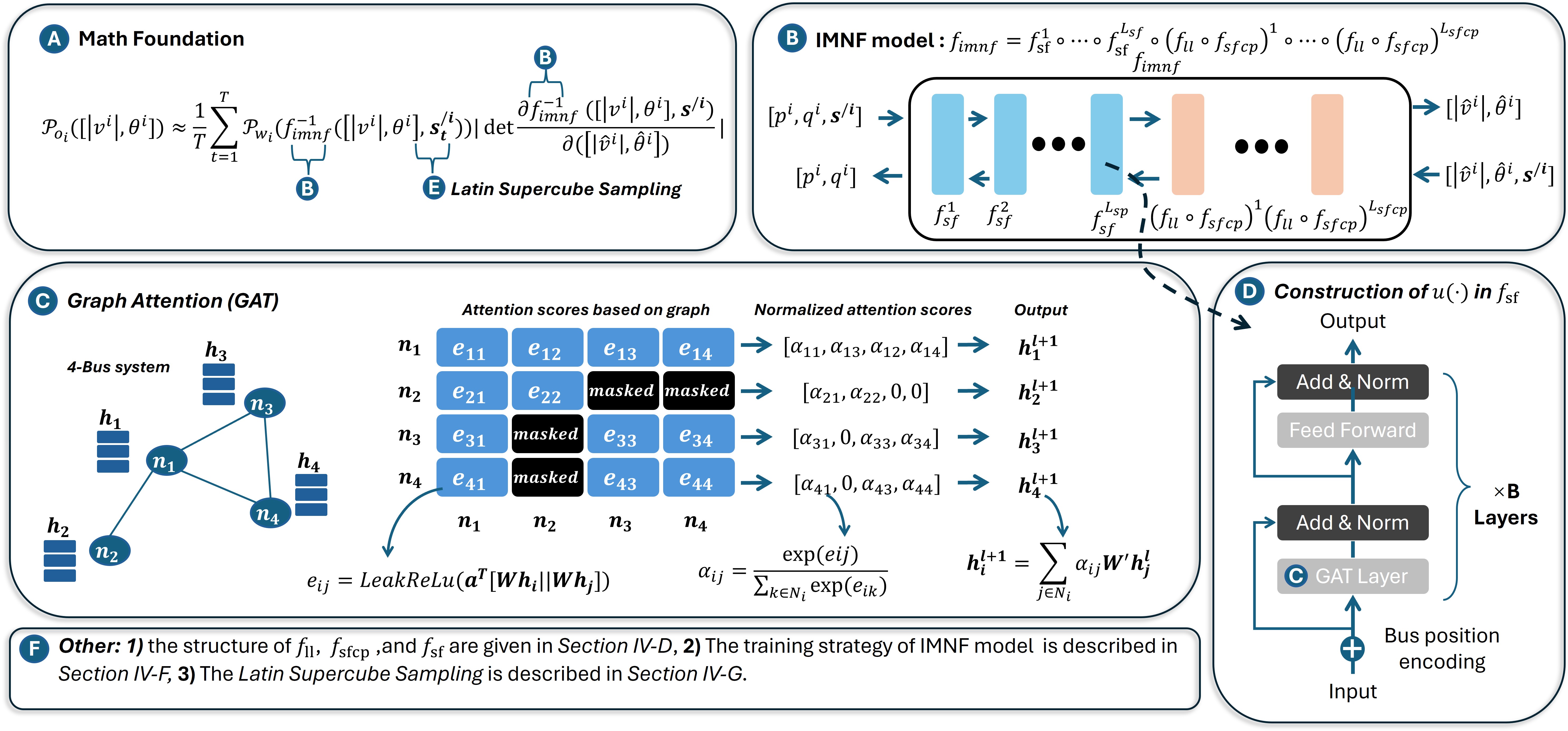}

    \caption{Overview of the proposed probabilistic power flow framework.
(A) Mathematical foundation of the density estimation.
(B) Construction of the IMNF model.
(C) GAT for topology-aware feature aggregation.
(D) Transformer-style block structure used within $f_{sp}$.
(E) LSS for efficient numerical integration.
}

    \label{fig:methodsummary}
\end{figure*}

\vspace{-0.1cm}

However, directly applying Eq.~(\ref{eq:cov_general}) requires evaluating the high-dimensional joint voltage distribution of the entire system, whereas PPF typically perfer the voltage distribution at each particular bus. To address this mismatch, let
\(\mathbf{s}^{/i}=[\mathbf{p}^{/i},\mathbf{q}^{/i}]\) collect the power injections at all buses except bus \(i\), and, for each fixed feasible \(\mathbf{s}^{/i}\), define the conditional PF mapping
\begin{equation}
g_{i,\mathbf{s}^{/i}}([p^i,q^i])
:=
\mathbf{E}_i\, f([p^i,q^i],\mathbf{s}^{/i})
=
[|v^i|,\theta^i],
\label{eq:conditional_pf_mapping}
\end{equation}
where \(\mathbf{E}_i \in \{0,1\}^{2\times 2N}\) is a selection matrix that extracts the voltage magnitude and phase angle at bus \(i\) from the output of the full PF mapping \(f\). As stated in Section~\ref{ProblemFormulation}, we restrict our attention to the feasible operating regime under normal steady-state conditions. Within this regime, the Jacobian of the conditional PF mapping \(g_{i,\mathbf{s}^{/i}}\) can be derived from the full PF Jacobian, as formalized in the following proposition.

\begin{proposition}[Conditional Power Flow Jacobian]\label{prop:cond_jacobian}
Let \(J=\partial(\mathbf{p},\mathbf{q})/\partial(|\mathbf{v}|,\boldsymbol{\theta})\) denote the full PF Jacobian and partition it according to bus \(i\) and the remaining buses as
\begin{equation}
\begin{bmatrix} d[p^i,q^i] \\ d\mathbf{s}^{/i} \end{bmatrix}
=
\begin{bmatrix} A & B \\ C & D \end{bmatrix}
\begin{bmatrix} d[|v^i|,\theta^i] \\ d\mathbf{u}^{/i} \end{bmatrix},
\label{eq:jacobian_partition}
\end{equation}
where \(\mathbf{u}^{/i}=[|\mathbf{v}|^{/i},\boldsymbol{\theta}^{/i}]\), and \(A,B,C,D\) are the corresponding blocks of \(J\). If \(D\) and \(A-BD^{-1}C\) are nonsingular, then the Jacobian of the conditional PF mapping in Eq.~(\ref{eq:conditional_pf_mapping}) is
\begin{equation}
\frac{\partial g_{i,\mathbf{s}^{/i}}([p^i,q^i])}
{\partial[p^i,q^i]}
=
(A-BD^{-1}C)^{-1}.
\label{eq:cond_jacobian}
\end{equation}
Consequently, \(g_{i,\mathbf{s}^{/i}}\) is locally bijective at the considered operating point.
\end{proposition}
\begin{proof}
By definition,
\begin{equation}
A = \frac{\partial [p^i,q^i]}{\partial[|v^i|,\theta^i]}, \;\;
B = \frac{\partial [p^i,q^i]}{\partial \mathbf{u}^{/i}}, \;\;
C = \frac{\partial \mathbf{s}^{/i}}{\partial[|v^i|,\theta^i]}, \;\;
D = \frac{\partial \mathbf{s}^{/i}}{\partial \mathbf{u}^{/i}}.
\end{equation}
Because all other injections are held fixed, \(d\mathbf{s}^{/i}=0\). The second block row of Eq.~(\ref{eq:jacobian_partition}) therefore gives
\(d\mathbf{u}^{/i}=-D^{-1}C\,d[|v^i|,\theta^i]\). Substituting this expression into the first block row yields
\begin{equation}
d[p^i,q^i]
=
(A-BD^{-1}C)\,d[|v^i|,\theta^i].
\end{equation}
Inverting this relation gives Eq.~(\ref{eq:cond_jacobian}). 
\end{proof}

Proposition~\ref{prop:cond_jacobian} allows Theorem~\ref{thm:cov} to be applied to \(g_{i,\mathbf{s}^{/i}}\) for every feasible conditioning scenario. Marginalizing the resulting conditional density over \(\mathbf{s}^{/i}\) yields the following result.

\begin{proposition}[Adjusted Change of Variable Theorem for PPF]\label{prop:ppf_cov}
Let \(\mathbf{s}^{/i}\) be distributed according to \(\mathcal{P}_{\mathbf{w}_{/i}}(\cdot)\), and let \(\{\mathbf{s}^{/i}_t\}_{t=1}^{T}\) be \(T\) samples drawn from this distribution. For a target voltage state \(\mathbf{o}_i=[|v^i|,\theta^i]\), define
\(\mathbf{x}_{i,t}=g_{i,\mathbf{s}^{/i}_t}^{-1}(\mathbf{o}_i)\), where \(\mathbf{x}_{i,t}=[p^i_t,q^i_t]\). Under the conditional regularity assumptions stated above, its marginal density is approximated by
\begin{subequations}\label{eq:prop_thm}
\begin{equation}
\begin{aligned}
\mathcal{P}_{\mathbf{o}_i}(\mathbf{o}_i)
\approx\
&\frac{1}{T}\sum_{t=1}^{T}
\mathcal{P}_{\mathbf{w}_i}(\mathbf{x}_{i,t} \mid \mathbf{s}^{/i}_t) \\
&\times
\left|
\det\!\left(
\left.
\frac{\partial g_{i,\mathbf{s}^{/i}_t}(\mathbf{x}_i)}
{\partial\mathbf{x}_i}
\right|_{\mathbf{x}_i=\mathbf{x}_{i,t}}
\right)
\right|^{-1},
\end{aligned}
\label{eq:1bus_marginal1}
\end{equation}
or equivalently, in explicit PDF form,
\begin{equation}
\begin{aligned}
\mathcal{P}_{\mathbf{o}_i}(\mathbf{o}_i)
\approx\
&\frac{1}{T}\sum_{t=1}^{T}
\mathcal{P}_{\mathbf{w}_i}\!\left(g_{i,\mathbf{s}^{/i}_t}^{-1}(\mathbf{o}_i) \mid \mathbf{s}^{/i}_t\right) \\
&\times
\left|
\det\!\left(
\frac{
\partial g_{i,\mathbf{s}^{/i}_t}^{-1}(\mathbf{o}_i)
}{
\partial\mathbf{o}_i
}
\right)
\right|.
\end{aligned}
\label{eq:1bus_marginal2}
\end{equation}
\end{subequations}
\end{proposition}
\begin{proof}
For each fixed feasible \(\mathbf{s}^{/i}\), if \(g_{i,\mathbf{s}^{/i}}\) has a nonsingular conditional Jacobian on the considered operating domain, the change-of-variables theorem can be applied to \(g_{i,\mathbf{s}^{/i}}\). Hence,
\begin{equation}
\mathcal{P}_{\mathbf{o}_i}(\mathbf{o}_i \mid \mathbf{s}^{/i})
=
\mathcal{P}_{\mathbf{w}_i}\!\left(
g_{i,\mathbf{s}^{/i}}^{-1}(\mathbf{o}_i)
\mid \mathbf{s}^{/i}
\right)
\left|
\det\!\left(
\frac{
\partial g_{i,\mathbf{s}^{/i}}^{-1}(\mathbf{o}_i)
}{
\partial\mathbf{o}_i
}
\right)
\right|.
\label{eq:1buscondition}
\end{equation}

The law of total probability then gives the marginal density
\begin{equation}
\begin{aligned}
\mathcal{P}_{\mathbf{o}_i}(\mathbf{o}_i)
&=
\int
\mathcal{P}_{\mathbf{w}_{/i}}(\mathbf{s}^{/i})
\,\mathcal{P}_{\mathbf{w}_i}\!\left(
g_{i,\mathbf{s}^{/i}}^{-1}(\mathbf{o}_i)
\mid\mathbf{s}^{/i}
\right)
\\
&\quad \times
\left|
\det\!\left(
\frac{
\partial g_{i,\mathbf{s}^{/i}}^{-1}(\mathbf{o}_i)
}{
\partial\mathbf{o}_i
}
\right)
\right|
\, \mathrm{d}\mathbf{s}^{/i}.
\end{aligned}
\label{eq:integration}
\end{equation}

Approximating Eq.~(\ref{eq:integration}) using the samples
\(\{\mathbf{s}^{/i}_t\}_{t=1}^{T}\) yields Eq.~(\ref{eq:1bus_marginal2}). The inverse-function identity
\begin{equation}
\left|
\det\!\left(
\frac{\partial g_{i,\mathbf{s}^{/i}}^{-1}(\mathbf{o}_i)}
{\partial\mathbf{o}_i}
\right)
\right|
=
\left.
\left|
\det\!\left(
\frac{\partial g_{i,\mathbf{s}^{/i}}(\mathbf{x}_i)}
{\partial\mathbf{x}_i}
\right)
\right|^{-1}
\right|_{\mathbf{x}_i=g_{i,\mathbf{s}^{/i}}^{-1}(\mathbf{o}_i)}
\end{equation}
then gives the equivalent form in Eq.~(\ref{eq:1bus_marginal1}).
\end{proof}

Eq.~(\ref{eq:1bus_marginal1}) is analogous to conditional normalizing flows~\cite{winkler2019learning, rasul2021multivariate}, with the remaining-bus injections \(\mathbf{s}^{/i}\) serving as the conditioning variables and the conditional Jacobian obtained from the physical PF Jacobian through Proposition~\ref{prop:cond_jacobian}.


As the probability density of the power injections
\(\mathcal{P}_{\mathbf{w}}(\cdot)\) is assumed to be known, and \textit{the conditional mapping \(g_{i,\mathbf{s}^{/i}}([p^i,q^i])\) (or \(\mathbf{E}_i\, f([p^i,q^i],\mathbf{s}^{/i})\)) and its determinant can be learned directly from PF data using an INN}. Substituting the learned conditional mapping and its Jacobian determinant into Eq.~\eqref{eq:1bus_marginal2} provides a tractable way to estimate the voltage distribution of an individual bus without full-system MC-based PPF after training. Although the scenarios \(\{\mathbf{s}^{/i}_t\}\) are sampled for marginalization, no PF computation is required. Instead, each density contribution is evaluated directly using the learned conditional mapping and its Jacobian determinant.  

Additionally, Eq.~\eqref{eq:1bus_marginal2} is an average over independently sampled scenarios, so the \(T\) density contributions can be computed independently and evaluated in parallel as a single batch, which is well suited to GPU parallelization and can potentially yield a substantial hardware-level speedup.

\subsection{Proposed Framework Overview}

Fig.~\ref{fig:methodsummary} presents a comprehensive overview of the proposed framework. 
The overall methodology is grounded in the probabilistic formulation derived in 
Eq.~(\ref{eq:1bus_marginal2}), introduced in Section~\ref{sec:math}. At its core, the IMNF model learns the conditional PF mapping, as detailed in Section~\ref{sec:imnf}. To capture spatial dependencies imposed by  the physical network topology, a GAT is incorporated, as  described in Section~\ref{sec:gat}. The training strategy of the IMNF model is explained  in Section~\ref{sec:trainingdesign}, while the sampling procedure based on Latin  Supercube Sampling (LSS) is presented in Section~\ref{sec:lss}. All implementation details and reproducible code are publicly available in our GitHub repositories\footnote{
The personal repository is available at 
\href{https://github.com/xiaweijie1996/Prob-for-PPF}{Personal GitHub Repository}, 
and the TU Delft repository is available at 
\href{https://github.com/distributionnetworksTUDelft/Invertible_NNs_for_PPF}{TU Delft GitHub Repository}.
}.


\subsection{{Invertible Mixed Neural Flow (IMNF)}}\label{sec:imnf}
To model conditional PF mapping, we propose IMNF, a hybrid INN-based model built by stacking three elementary invertible layers: a \textit{Linear layer}, an \textit{Exp layer}, and a \textit{Spline layer}~\cite{durkan2019neural}. For bus $i$, all three layers act on the vector $\mathbf{x}^{i} = [p^{i}, q^{i}]$, conditioned on the scenario $\mathbf{s}^{/i}$ collecting the power injections at every other bus. We write $\mathbf{x}^{i}_{l}$ for the state of this vector after $l$ layers, so the model input is $\mathbf{x}^{i}_{0} = [p^{i}_0, q^{i}_0] = [p^{i}, q^{i}]$. Each layer below is described by its forward map, its inverse map, and its Jacobian, since these three ingredients are exactly what is needed to evaluate Eq.~(\ref{eq:1bus_marginal2}).

\subsubsection{Linear Layer}\label{sec:linear}
The linear layer, denoted $f_{\mathrm{ll}}(\cdot)$, applies a trainable, invertible $2\times2$ matrix $A$ to the input vector, with no dependence on $\mathbf{s}^{/i}$:
\begin{align}
& \text{Forward:} & [p^{i}_1, q^{i}_1] &= f_{\mathrm{ll}}([p^{i}_0, q^{i}_0]) = A\,[p^{i}_0, q^{i}_0], \\
& \text{Inverse:} & [p^{i}_0, q^{i}_0] &= f_{\mathrm{ll}}^{-1}([p^{i}_1, q^{i}_1]) = A^{-1}[p^{i}_1, q^{i}_1], \\
& \text{Jacobian:} & J_{f_{\mathrm{ll}}} &= A, \qquad \left|\det J_{f_{\mathrm{ll}}}\right| = \left|\det A\right|.
\end{align}
Because $A$ does not depend on the input, its Jacobian is simply $A$ itself, and the determinant is constant across the domain.

\subsubsection{Exp Layer}\label{sec:exp}
The Exp layer, denoted $f_{\mathrm{sfcp}}(\cdot)$ (following the naming of FCPFlow~\cite{xia2025flow}), is an affine coupling transformation. Given an input $[p^{i}_0, q^{i}_0]$ and condition $\mathbf{s}^{/i}$,
\begin{align}
& \text{Forward:} & q^{i}_1 &= q^{i}_0 \odot \exp\!\bigl(s_1(p^{i}_0, \mathbf{s}^{/i})\bigr) + t_1(p^{i}_0, \mathbf{s}^{/i}), \\
& & p^{i}_1 &= p^{i}_0 \odot \exp\!\bigl(s_2(q^{i}_1, \mathbf{s}^{/i})\bigr) + t_2(q^{i}_1, \mathbf{s}^{/i}), \\
& \text{Inverse:} & p^{i}_0 &= \bigl(p^{i}_1 - t_2(q^{i}_1, \mathbf{s}^{/i})\bigr) \odot \exp\!\bigl(-s_2(q^{i}_1, \mathbf{s}^{/i})\bigr), \\
& & q^{i}_0 &= \bigl(q^{i}_1 - t_1(p^{i}_0, \mathbf{s}^{/i})\bigr) \odot \exp\!\bigl(-s_1(p^{i}_0, \mathbf{s}^{/i})\bigr),
\end{align}
so that $[p^{i}_1, q^{i}_1] = f_{\mathrm{sfcp}}([p^{i}_0, q^{i}_0], \mathbf{s}^{/i})$, where $s_1(\cdot), t_1(\cdot), s_2(\cdot), t_2(\cdot)$ can be any NNs. The forward map is a composition of two coupling half-steps, each updating one coordinate while holding the other fixed as the conditioner: $h : (p^{i}_0, q^{i}_0) \mapsto (p^{i}_0, q^{i}_1)$, followed by $g : (p^{i}_0, q^{i}_1) \mapsto (p^{i}_1, q^{i}_1)$, so that $f_{\mathrm{sfcp}} = g \circ h$. Each half-step accordingly has a triangular Jacobian,
\begin{equation}
J_h =
\begin{bmatrix}
1 & 0 \\[4pt]
\dfrac{\partial q^{i}_1}{\partial p^{i}_0} & \exp\!\bigl(s_1(p^{i}_0,\mathbf{s}^{/i})\bigr)
\end{bmatrix}\!, \;\;
J_g =
\begin{bmatrix}
\exp\!\bigl(s_2(q^{i}_1,\mathbf{s}^{/i})\bigr) & \dfrac{\partial p^{i}_1}{\partial q^{i}_1} \\[4pt]
0 & 1
\end{bmatrix}\!,
\end{equation}
with $\det J_h = \exp(s_1(p^{i}_0,\mathbf{s}^{/i}))$ and $\det J_g = \exp(s_2(q^{i}_1,\mathbf{s}^{/i}))$. By the chain rule, the Jacobian of the full Exp layer is $J_{f_{\mathrm{sfcp}}} = J_g J_h$; this product is not itself triangular, since the $(1,2)$ entry of $J_g$ and the $(2,1)$ entry of $J_h$ are generally nonzero, but its determinant is still the product of the two half-step determinants,
\begin{equation}
\left|\det J_{f_{\mathrm{sfcp}}}\right| = \left|\det J_g\right|\left|\det J_h\right| = \exp\!\bigl(s_1(p^{i}_0, \mathbf{s}^{/i}) + s_2(q^{i}_1, \mathbf{s}^{/i})\bigr).
\end{equation}
\subsubsection{Spline Layer}\label{sec:sf}
The spline layer, denoted $f_{\mathrm{sf}}(\cdot)$, replaces the affine coupling of the Exp layer with an elementwise monotonic rational–quadratic spline. For the detailed forward/inverse construction and the closed-form derivatives, we refer the reader to~\cite{durkan2019neural}.

\subsubsection{Stacking the Linear, Exp, and Spline Layers}
In the previous sections, we introduced the three elementary layers $f_{\mathrm{ll}}(\cdot)$, $f_{\mathrm{sfcp}}(\cdot)$, and $f_{\mathrm{sf}}(\cdot)$, together with their inverses $f_{\mathrm{ll}}^{-1}(\cdot)$, $f_{\mathrm{sfcp}}^{-1}(\cdot)$, and $f_{\mathrm{sf}}^{-1}(\cdot)$. The proposed IMNF model $f_{\mathrm{imnf}}(\cdot)$ is constructed by stacking these layers: the linear and Exp layers are paired and alternated $L_{\mathrm{sfcp}}$ times, and the result is followed by $L_{\mathrm{sf}}$ Spline layers, i.e.,
\begin{equation}
    f_{\mathrm{imnf}}
=
f_{\mathrm{sf}}^{1}
\circ \cdots \circ
f_{\mathrm{sf}}^{L_{\mathrm{sf}}}
\circ
f_{\mathrm{ll}}^{1}
\circ
f_{\mathrm{sfcp}}^{1}
\circ \cdots \circ
f_{\mathrm{ll}}^{L_{\mathrm{sfcp}}}
\circ
f_{\mathrm{sfcp}}^{L_{\mathrm{sfcp}}}(\cdot),
\label{eq:imnf_compose}
\end{equation}
where $L_{\mathrm{sf}}$ and $L_{\mathrm{sfcp}}$ denote the number of stacked Spline layers and linear-Exp pairs, respectively, and $\circ$ represents function composition, i.e., $(f \circ g)(x) = f(g(x))$. Since every layer in Eq.~(\ref{eq:imnf_compose}) is invertible with a closed-form Jacobian (Sections~\ref{sec:linear}--\ref{sec:sf}), and the determinant of a Jacobian is multiplicative under composition, applying this property recursively across all layers shows that the determinant of $J_{f_{\mathrm{imnf}}}$ is simply the product of the per-layer determinants,
\begin{equation}
\left|\det J_{f_{\mathrm{imnf}}}\right|
=
\prod_{l=1}^{L_{\mathrm{sf}}} \left|\det J_{f_{\mathrm{sf}}^{l}}\right|
\;\times\;
\prod_{l=1}^{L_{\mathrm{sfcp}}} \left|\det J_{f_{\mathrm{ll}}^{l}}\right|\left|\det J_{f_{\mathrm{sfcp}}^{l}}\right|,
\label{eq:imnf_det}
\end{equation}
which can be accumulated analytically during the forward pass at no extra computational cost, and is the quantity used in the density expression of Eq.~(\ref{eq:1bus_marginal2}) and in the Jacobian-consistency loss of Section~\ref{sec:trainingdesign}.


\subsection{Graph Attention Networks (GAT)}\label{sec:gat}

To better exploit the network topology and improve model performance, we use GATs~\cite{wu2022graph,velivckovic2017graph}, rather than standard feedforward NNs, to parameterize the neural subnetworks in the proposed model, such as \(s_l(\cdot)\) and \(t_l(\cdot)\) in \(f_{\mathrm{sfcp}}(\cdot)\) and the subnetworks in \(f_{\mathrm{sp}}(\cdot)\).

GAT injects topological priors by enforcing a \textit{masked attention} mechanism, where each token attends only to its neighbors as defined by the network topology. This restricts the computation of attention coefficients to adjacent tokens, making the model more physically meaningful for power systems. The GAT can be expressed as
\begin{equation}
     \text{GAT:} \quad 
    \alpha_{ij} = \text{softmax}_j(e_{ij})
    = \frac{\exp(e_{ij})}{\sum_{k \in \mathcal{N}_i} \exp(e_{ik})}, \\
\end{equation}
where $\mathcal{N}_i$ denotes the neighboring tokens of token $i$, and $e_{ij}$ is the unnormalized attention score between tokens $i$ and $j$, defined as
\begin{equation}
e_{ij} = 
\text{LeakyReLU}
\left(
\mathbf{a}^{\top}
\left[
\mathbf{W}\mathbf{h}_i
 \Vert 
\mathbf{W}\mathbf{h}_j
\right]
\right),
\end{equation}
with $\mathbf{h}_i$ and $\mathbf{h}_j$ denoting the feature vectors of tokens $i$ and $j$, $\mathbf{W}$ a learnable weight matrix, $\mathbf{a}$ a learnable attention vector, and $\Vert$ representing concatenation. We use absolute positional encoding to identify each bus (token, see Fig.~\ref{fig:methodsummary}(D)), so a single GAT model is shared across all buses and used for both training and inference, rather than training a separate model per bus.
\vspace{-0.1cm}
\subsection{Training Design}\label{sec:trainingdesign}
To train the proposed IMNF model, we adopt the mean squared error (MSE) as the primary loss function. Moreover, to fully exploit the physical property of the PF, the inherent invertibility between $(\mathbf{p}, \mathbf{q})$ and  $(|\mathbf{v}|, \boldsymbol{\theta})$, we introduce a \textit{bidirectional training} strategy.  This training mechanism enforces consistency in both the forward and inverse transformation, and we find empirically that it significantly improves robustness and accelerates convergence. The procedure is summarized in Algorithm~\ref{alg:imnf_training}. Formally, the forward and inverse losses at bus \(i\), conditioned on \(\mathbf{s}^{/i}\), are defined as
\begin{subequations}\label{eq:bidirectional_losses}
\begin{align}
    \text{loss}_{w2o} &=
    \mathrm{MSE}\!\left(
    f_{\mathrm{imnf}}([p^i,q^i],\mathbf{s}^{/i}),
    [|v^i|,\theta^i]
    \right),
    \label{eq:loss_w2o}\\
    \text{loss}_{o2w} &=
    \mathrm{MSE}\!\left(
    f_{\mathrm{imnf}}^{-1}([|v^i|,\theta^i],\mathbf{s}^{/i}),
    [p^i,q^i]
    \right).
    \label{eq:loss_o2w}
\end{align}
\end{subequations}
Beyond matching point predictions, the density in Eq.~(\ref{eq:1bus_marginal2}) also depends on the Jacobian determinant of $f_{\mathrm{imnf}(\cdot)}$, so an accurate point mapping does not by itself guarantee an accurate density estimate. We therefore also compare the model's predicted determinant $\left|\det J_{f_{\mathrm{imnf}}}\right|$, available in closed form via Eq.~(\ref{eq:imnf_det}), against the true conditional PF Jacobian determinant $\left|\det J^{i}_{\mathrm{PF}}\right|$ at bus $i$, computed from the PF simulation using Eq.~(\ref{eq:cond_jacobian}), through a third loss term,
\begin{equation}
    \text{loss}_{jac} =
    \mathrm{MSE}\!\left(
    \left|\det J_{f_{\mathrm{imnf}}}\right|,
    \left|\det J^{i}_{\mathrm{PF}}\right|
    \right).
\end{equation}
The overall training objective is a weighted combination of the three terms, which is expressed as
\begin{equation}
    \text{loss} =
    \omega\, \text{loss}_{w2o}
    + (1-\omega)\, \text{loss}_{o2w}
    + \lambda\, \text{loss}_{jac},
      \omega \in [0,1],\ \lambda \geq 0.
    \label{eq:final_loss}
\end{equation}

\begin{algorithm}[t]
\caption{Training of IMNF with Bidirectional PF Constraints}
\label{alg:imnf_training}
\begin{algorithmic}[1]

\Require
\Statex \hspace{1em} Injected power distribution $\mathcal{P}_{\mathbf{w}}([\mathbf{p}, \mathbf{q}])$
\Statex \hspace{1em} Bidirectional model $f_{\mathrm{imnf}}(\cdot)$
\Statex \hspace{1em} Weights $\omega \in [0,1]$, $\lambda \geq 0$
\Statex \hspace{1em} Distribution system with $N$ buses

\ForAll{minibatches}
    \State Sample $[\mathbf{p}, \mathbf{q}]$ from $\mathcal{P}_{\mathbf{w}}([\mathbf{p}, \mathbf{q}])$.
    \State Sample index $i \in \{1,\ldots,N\}$
    \State Get $(p^i, q^i, \mathbf{s}^{/i}, |v^i|, \theta^i, \left|\det J^{i}_{\mathrm{PF}}\right|)$ via PF simulation

    \State \textbf{(1) Forward transformation}
    \State $(\widehat{|v|}^i,\, \widehat{\theta}^i)
            \gets f_{\mathrm{imnf}}([p^i,q^i],\mathbf{s}^{/i})$
    \State $\text{loss}_{w2o, i} \gets
        \mathrm{MSE}\!\left(
        [\widehat{|v|}^i,\, \widehat{\theta}^i],
        [|v^i|,\, \theta^i]
        \right)$

    \State \textbf{(2) Inverse transformation}
    \State $(\widehat{p}^i,\, \widehat{q}^i)
            \gets f_{\mathrm{imnf}}^{-1}([|v^i|,\theta^i],\mathbf{s}^{/i})$
    \State $\text{loss}_{o2w, i} \gets
        \mathrm{MSE}\!\left(
        [\widehat{p}^i,\, \widehat{q}^i],
        [p^i, q^i]
        \right)$

    \State \textbf{(3) Jacobian consistency}
    \State $\left|\widehat{\det J}^{i}\right| \gets \left|\det J_{f_{\mathrm{imnf}}}\right|$ via Eq.~(\ref{eq:imnf_det}), evaluated along the forward pass of step (1)
    \State $\text{loss}_{jac, i} \gets
        \mathrm{MSE}\!\left(
        \left|\widehat{\det J}^{i}\right|,
        \left|\det J^{i}_{\mathrm{PF}}\right|
        \right)$

    \State \textbf{(4) Total loss}
    \State $\text{loss}_i \gets
        \omega \cdot \text{loss}_{w2o, i}
        + (1 - \omega) \cdot \text{loss}_{o2w, i}
        + \lambda \cdot \text{loss}_{jac, i}$

    \State Update parameters of $f_{\mathrm{imnf}}$ using optimizer (Adam)
\EndFor
\end{algorithmic}
\end{algorithm}
\vspace{-0.3cm}
\subsection{Scenario Sampling}\label{sec:lss}
To compute the final density in Eq.~\eqref{eq:1bus_marginal2}, it is necessary to draw a set of scenarios $\{\mathbf{s}_t^{/i}\}_{t=1}^{T}$ from the joint distribution $\mathcal{P}_{\mathbf{w}_{/i}}(\cdot)$. To do this, we adopt LSS, a  scheme that combines quasi--Monte Carlo (QMC) point sets and Latin Hypercube Sampling (LHS). LSS  is particularly effective for high-dimensional integration~\cite{owen1998latin,hajian2012probabilistic}.

Let $\mathbf{s}^{/i}\in\mathbb{R}^{d}$ denote the scenario vector to be sampled, and let $\mathbf{u}\in[0,1]^d$ be its corresponding uniform representation. LSS first partitions the $d$ dimensions into $K$ disjoint groups,
\begin{align}
& \{1,\dots,d\} = G_1 \cup \cdots \cup G_K, \\
& G_a \cap G_b = \emptyset \ (a\neq b),\\
&|G_k|=d_k,\ \sum_{k=1}^{K} d_k = d,
\label{eq:lss_groups}  
\end{align}
and generates, for each group $k$, a low-discrepancy QMC point set $\mathbf{U}^{(k)}=\{\mathbf{u}^{(k)}_t\}_{t=1}^{T}\subset[0,1]^{d_k}$. To introduce Latin-style stratification across groups without destroying the within-group low-discrepancy structure, LSS applies a \emph{single shared permutation} $\pi_k$ of $\{1,\dots,T\}$ to all coordinates in group $k$, i.e.,
\begin{equation}
\tilde{\mathbf{u}}^{(k)}_t = \mathbf{u}^{(k)}_{\pi_k(t)},
\qquad t=1,\dots,T.
\label{eq:lss_permute}
\end{equation}
The full $d$-dimensional sample is then obtained by concatenation,
\begin{equation}
\mathbf{u}_t = \big[\tilde{\mathbf{u}}^{(1)}_t \,\|\, \tilde{\mathbf{u}}^{(2)}_t \,\|\, \cdots \,\|\, \tilde{\mathbf{u}}^{(K)}_t\big]\in[0,1]^d,
\qquad t=1,\dots,T.
\label{eq:lss_concat}
\end{equation}
Finally, $\mathbf{u}_t$ is mapped to the target scenario space via the inverse marginal transforms, yielding $\mathbf{s}^{/i}_t$.

\section{Experiments}\label{sec:allexp}

\subsection{Experimental Setup for PF Simulation}


In this section, we first evaluate IMNF as a full-system PF solver in both the forward and inverse directions. In this setting, the conditioning vector is omitted, i.e., \(\mathbf{s}^{/i}=\emptyset\), and the model represents $
f([\mathbf{p},\mathbf{q}])
=
[|\mathbf{v}|,\boldsymbol{\theta}]
$
and
$
f^{-1}([|\mathbf{v}|,\boldsymbol{\theta}])
=
[\mathbf{p},\mathbf{q}].
$
Reference input--output pairs obtained from conventional PF calculations for MC-sampled injection scenarios are used as the ground truth.

We compare our method against several categories of benchmarks: 1) other INN-based architectures that are not specifically designed for PPF, 2) other NN–based PF solvers that do not rely on an invertible structure, and 3) IMNF variants in which $f_{\mathrm{sf}}(\cdot)$ is constructed using either FNNs or GATs, enabling an ablation-style comparison. All experiments are conducted on a 34-bus test system described in~\cite{duque2024tensor}. The evaluation metric is the mean absolute error (MAE), defined as
\begin{equation}
    \text{MAE} 
    = \frac{1}{N_{\text{samples}}}
      \sum_{n=1}^{N_{\text{samples}}} 
      |\mathbf{x}_n - \hat{\mathbf{x}}_n|,
\end{equation}
where $\mathbf{x}_n$ and $\hat{\mathbf{x}}_n$ denote the ground truth and predicted vectors, respectively. We use $N_{\text{samples}} = 10{,}000$ test samples for evaluation. For all models, we adopt the same hyperparameter configuration: a learning rate of $0.001$, a batch size of $64$,  $20{,}000$ training epochs, and $\omega$ in Eq.~(\ref{eq:final_loss}) is set as $0.5$.  The model sizes are kept comparable at approximately $7\times 10^5$ trainable parameters. Other experimental details can be found in our repository. All experiments are trained using the Adam optimizer on the same hardware environment (an NVIDIA A10 GPU), ensuring a fair comparison across methods.

\subsection{Experimental Results of  PF Simulation}

Table~\ref{tab:pfcomparison} summarizes the comparison results of all benchmarks and our proposed IMNF model. Examining the INN-based benchmarks first, we observe that despite similar performance in predicting $[\mathbf{p}, \mathbf{q}]$, both FCPFlow and SplineFlow significantly outperform the other INN models in predicting $[|\mathbf{v}|, \boldsymbol{\theta}]$. This empirical finding provides one of the motivations for selecting $f_{\mathrm{sf}}(\cdot)$ and $f_{\mathrm{sfcp}}(\cdot)$ as the fundamental building blocks of our proposed IMNF model.

Turning to the NN–based PF simulation benchmarks, we note that the physics-guided NN~\cite{hu2020physics} achieves the strongest performance. This model also leverages system topology information and incorporates the invertibility of PF. However, its invertibility is achieved in an indirect manner, in contrast to the explicit invertibility offered by the IMNF framework adopted in our work.

Lastly, our proposed IMNF model shows the best performance experimentally, combining the advantages of both FCPFlow and SplineFlow. With IMNF-GAT, further leveraging the GAT, it outperforms the IMNF-FNN, which is based solely on FNN.

\begin{table}[t]
\centering
\caption{Experimental results of PF simulation}
\label{tab:pfcomparison}
\begin{tabular}{|c|c|c|c|c|}
\hline
Method & MAE $\mathbf{P}$ & MAE $\mathbf{q}$ & MAE $\boldsymbol{\theta}$ & MAE $|\mathbf{v}|$ \\
\hline
\multicolumn{5}{|c|}{\textit{INN Benchmarks}} \\
\hline
Nice~\cite{dinh2014nice} & 0.3396 & 0.2289 & 0.1554 &  0.0758\\
Realnvp~\cite{dinh2016density} & 0.2906 & 0.2595 & 0.1338 & 0.0745\\
Tarflow~\cite{zhai2024normalizing} &  0.3151 & 0.2323 & 0.1597 & 0.0667\\
FCPFlow~\cite{xia2025flow} & 0.3449 & 0.4359 & \underline{0.0459} & 0.2555\\
SplineFlow~\cite{durkan2019neural} & 0.2979 & \underline{0.2236} & 0.1577 & \underline{0.0395} \\
\hline
\multicolumn{5}{|c|}{\textit{NN-based PF Simulation Benchmarks}} \\
\hline
GAECN~\cite{wu2022graph} & 0.4547 & 0.6177 & / & / \\
Physics-Guided NN~\cite{hu2020physics} & \underline{0.2785} & 0.4513 & / & / \\
PowerFlowNet~\cite{lin2024powerflownet} & 0.2886 & 0.5086 & / & / \\
\hline
\multicolumn{5}{|c|}{\textit{Our Proposed Method}} \\
\hline
IMNF-FNN & 0.1780 & 0.1761 & 0.0190 & 0.0335 \\
IMNF-GAT & \textbf{0.1489} & \textbf{0.1619} & \textbf{0.0151} & \textbf{0.0220} \\
\hline
\end{tabular}
\begin{tablenotes}[flushleft]
\footnotesize
\item Note: $x_n$ and $\hat{x}_n$ are normalized before MAE; "/" denotes unavailable outputs due to one-directional prediction of the corresponding benchmarks.
\end{tablenotes}
\vspace{-0.6cm}
\end{table}

\subsection{Experimental Setup for PPF Simulation}
In this section, we evaluate the performance of the proposed framework on density estimation for PPF. The proposed framework is compared against two PPF density-approximation benchmarks. The first benchmark is the linear method~\cite{wang2016analytical}, in which the PF is approximated as a linear mapping of the form $\mathbf{y} = A\mathbf{x} + \mathbf{b}$. The second benchmark is the piecewise linear method~\cite{gao2023analytical}, where the power injection space is segmented into multiple regimes and the PF is locally approximated by linear mappings within each regime, resulting in an overall nonlinear approximation. Following the setup described in Section~\ref{ProblemFormulation}, we assume that the power injection distribution $\mathcal{P}_{\mathbf{w}}(\cdot)$ is known from historical data and is modeled by a GMM. Under this assumption, the marginal and conditional densities $\mathcal{P}_{\mathbf{w}_{/i}}(\cdot)$ and $\mathcal{P}_{\mathbf{w}_{i}}(\cdot)$ appearing in Eq.~(\ref{eq:1bus_marginal2}) admit the following closed-form expressions
\begin{align}
&\mathcal{P}_{\mathbf{w}_{/i}}(\mathbf{s}_t^{/i})
= 
\sum_{k=1}^{K}
    \pi_k \,
    \mathcal{N}\!\left(
        \mathbf{s}_t^{/i}
        \,;
        \boldsymbol{\mu}^{/i}_k,\,
        \boldsymbol{\Sigma}^{/i}_k
    \right), \\[6pt]
&\mathcal{P}_{\mathbf{w}_i}\!\left(
    f_{\mathrm{imnf}}^{-1}\!\left(
        [\lVert v_t^{\,i}\rVert,\theta_t^{\,i}],\mathbf{s}_t^{/i}
    \right)
    \,\middle|\,
    \mathbf{s}_t^{/i}
\right)
\\
&\qquad=
\sum_{k=1}^{K}
    \pi_k(\mathbf{s}_t^{/i})\,
    \mathcal{N}\!\left(
        f_{\mathrm{imnf}}^{-1}\!\left(
            [\lVert v_t^{\,i}\rVert,\theta_t^{\,i}],\mathbf{s}_t^{/i}
        \right)
        \,\middle|\,
        \boldsymbol{\mu}^{i|/i}_k,\,
        \boldsymbol{\Sigma}^{i|/i}_k
    \right), \notag\\
&\pi_k(\mathbf{s}_t^{/i})
=
\frac{
    \pi_k\,
    \mathcal{N}\!\left(
        \mathbf{s}_t^{/i}
        \,\middle|\,
        \boldsymbol{\mu}^{/i}_k,\,
        \boldsymbol{\Sigma}^{/i}_k
    \right)
}{
    \sum_{j=1}^{K}
    \pi_j\,
    \mathcal{N}\!\left(
        \mathbf{s}_t^{/i}
        \,\middle|\,
        \boldsymbol{\mu}^{/i}_j,\,
        \boldsymbol{\Sigma}^{/i}_j
    \right)
}, \label{gmmconditionalweight}
\end{align}
where $\boldsymbol{\mu}^{/i}_k$ and $\boldsymbol{\Sigma}^{/i}_k$ denote the marginal mean and covariance of the $k$-th GMM component with the $i$-th bus removed, $\pi_k$ is the weight of $k$-th component. The conditional mean $\boldsymbol{\mu}^{i|/i}_k$, the conditional covariance $\boldsymbol{\Sigma}^{i|/i}_k$, and the conditional mixture weights $\pi_k(\mathbf{s}_t^{/i})$ follow directly from the standard conditional Gaussian and conditional mixture formulations. 

For evaluation, we first perform MC–based PF simulations with $2{,}500$ samples. For each evaluated bus, a GMM is fitted to the resulting voltage samples and treated as the reference (\textit{ground-truth}) distribution with a closed-form PDF, denoted by $\{\mathcal{P}_{\mathbf{o}_i}(\cdot)\}_{i=1}^{N_s}$. Regarding the evaluation metrics, we quantify the discrepancy between the estimated and reference voltage distributions using the Jensen--Shannon divergence (JSD) and the total variation distance (TVD). Given reference probability distribution $p(\mathbf{x})$ and predicted distribution $q(\mathbf{x})$ defined on $\mathbf{x}\in\mathbb{R}^2$, the TVD is defined as
\begin{equation}
\mathrm{TVD}(p,q)
=
\frac{1}{2}
\int
\left| p(\mathbf{x}) - q(\mathbf{x}) \right| \mathrm{d}\mathbf{x}.
\label{eq:tv_distance}
\end{equation}
The JSD is a symmetric and bounded information-theoretic measure defined as
\begin{equation}
\mathrm{JSD}(p,q)
=
\frac{1}{2}\,\mathrm{KL}\!\left(p \,\middle\|\, m\right)
+
\frac{1}{2}\,\mathrm{KL}\!\left(q \,\middle\|\, m\right), \,
m = \frac{1}{2}(p+q),
\label{eq:js_divergence}
\end{equation}
where $\mathrm{KL}(\cdot\|\cdot)$ denotes the Kullback--Leibler divergence,
\begin{equation}
\mathrm{KL}(p\|q)
=
\int
p(\mathbf{x})
\log\frac{p(\mathbf{x})}{q(\mathbf{x})}
\, \mathrm{d}\mathbf{x}.
\label{eq:kl_divergence}
\end{equation}
Smaller values of JSD and TVD indicate greater similarity between the two distributions.

For the training of the IMNF model, we adopt a circular learning-rate schedule with a maximum learning rate of $5\times10^{-4}$ and a minimum learning rate of $5\times10^{-6}$. The same weighting coefficient $\omega$ in~\eqref{eq:final_loss} is fixed to $0.5$. All experiments are trained using the AdamW optimizer on identical hardware, specifically an NVIDIA A100 GPU. The model sizes are approximately $2 \sim 8\times10^{6}$ parameters based on the system size. We use a batch size of $240$. The power injection distribution $\mathcal{P}_{\mathbf{w}(\cdot)}$ is modeled as a GMMs, initialized using historical data. Here, we do not set a fixed number of training epochs as we aim for the best model. 

\begin{figure*}[t]
    \centering
    \includegraphics[width=1\linewidth]{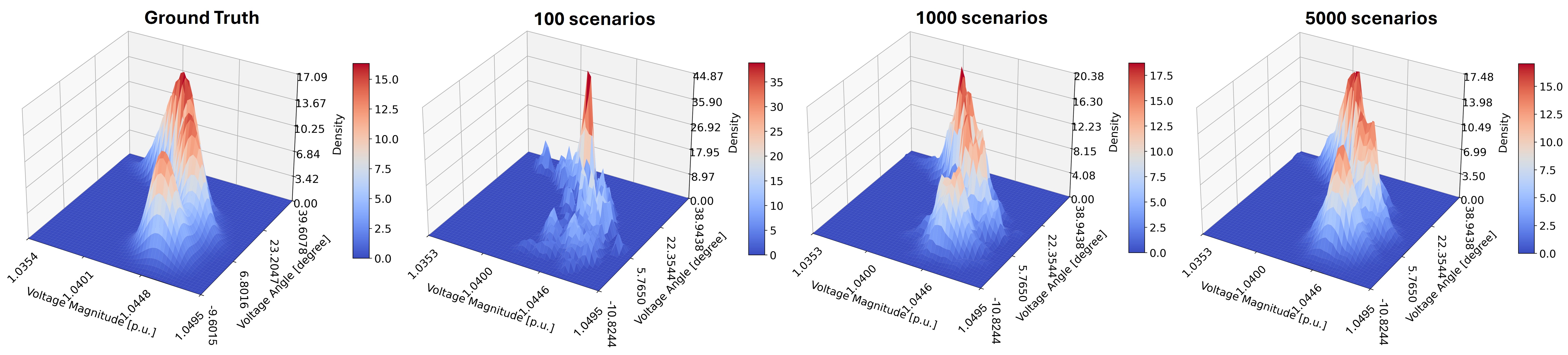}
    \caption{Example of voltage density estimation using the proposed IMNF model for the 39-bus system at bus~1 using LSS. Increasing the number of scenarios yields progressively more accurate density estimates.}
    \label{fig:increasingscenarios}
    \vspace{-0.5cm}
\end{figure*}

\begin{figure*}[htp]
    \centering
    \subfloat[The voltage distribution of $17$-th bus of IEEE-39 bus system.]{
        \includegraphics[width=\linewidth]{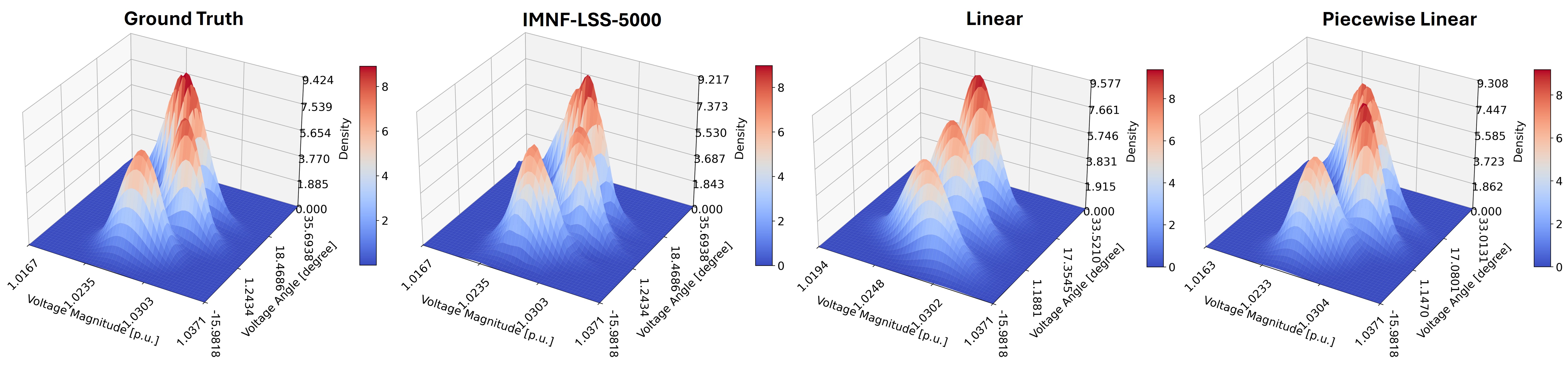}
    }\par\medskip

    \subfloat[The voltage distribution of $2$-th bus of CIGRE-HV system.]{
        \includegraphics[width=\linewidth]{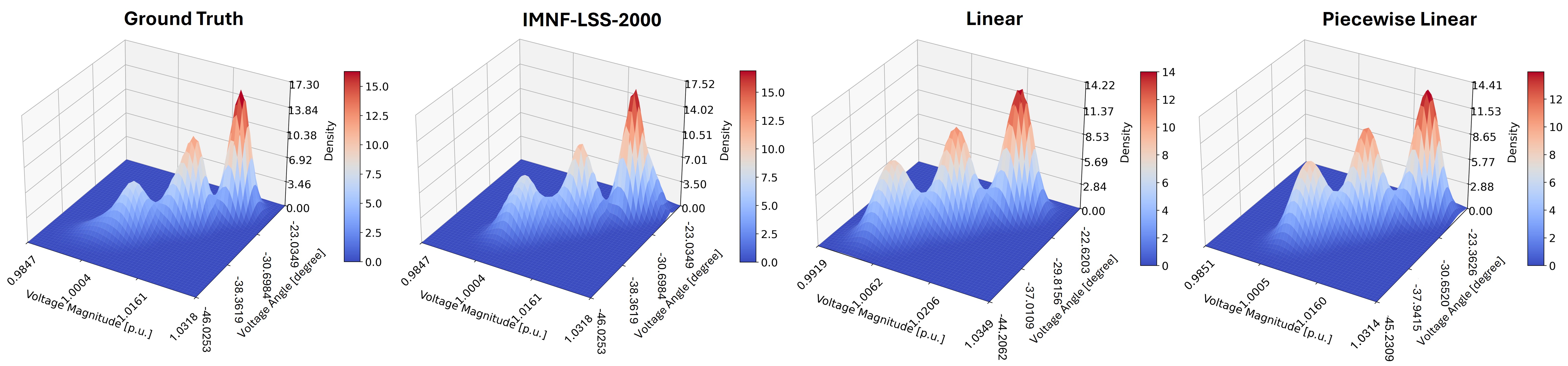}
    }\par\medskip

    \subfloat[The voltage distribution of $41$-th bus of IEEE 69 bus system.]{
        \includegraphics[width=\linewidth]{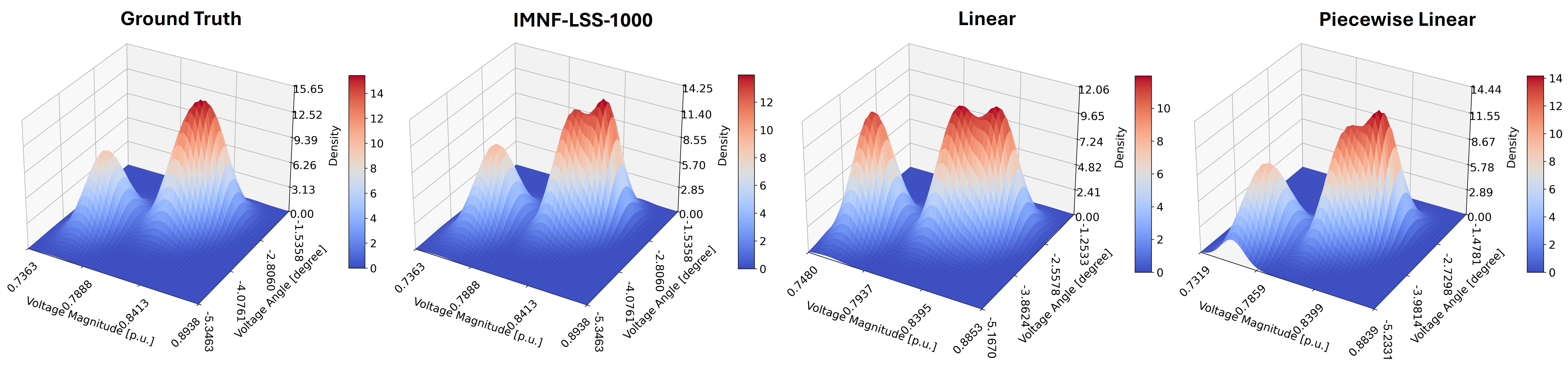}
    }\par\medskip

    \subfloat[The voltage distribution of $7$-th bus of CIGRE-LV system.]{
        \includegraphics[width=\linewidth]{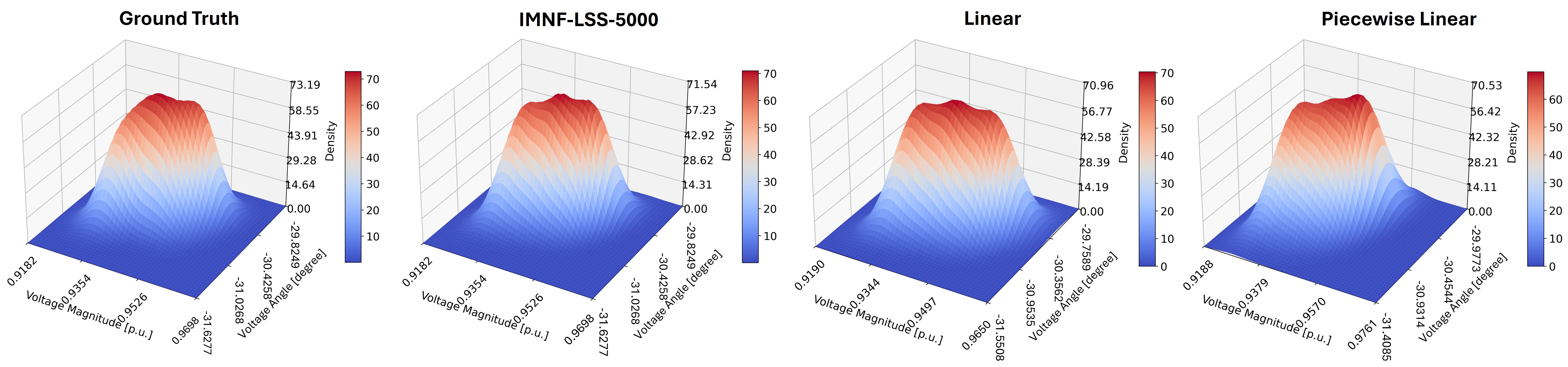}
    }
    \caption{Exemplar density evaluation results from the proposed framework and benchmarks.}
    \label{fig:four_panels}
\end{figure*}

\begin{table}[t]
\centering
\caption{Experimental results of PPF simulation}
\label{table:expppfcomparision}
\begin{threeparttable}
\begin{tabular}{|c|c|c|c|c|c|}
\hline
\textbf{Method} & JSD & TVD &  Method & JSD & TVD \\
\hline
 & \multicolumn{5}{c|}{\textit{Transmission system}} \\
\hline
 & \multicolumn{2}{c|}{\textit{IEEE 39-Bus system}} & -- & \multicolumn{2}{c|}{\textit{CIGRE-HV}~\cite{barsali2014benchmark}} \\
\hline
Linear~\cite{wang2016analytical} & 0.2982 & 0.0964 & Linear & \underline{0.1148} & \textbf{0.0145} \\
PLinear~\cite{gao2023analytical} & 0.1113 & 0.0148 & PLinear & 0.1291 & 0.0194\\
MC-500 & 0.1469 & 0.0271 & MC-10 & 0.2605 & 0.0664 \\
MC-1000 & 0.1130 & 0.0161 & MC-50 & 0.2062 & 0.0337\\
MC-5000  & 0.0674 & 0.0056 & MC-2000 & 0.1217 &	0.0167\\
LSS-500 & 0.1461 & 0.0278 & LSS-10  & 0.3436 & 0.1115\\
LSS-1000 & 0.1093 & 0.0159 & \underline{LSS-50} & 0.1202 & 0.0184 \\
LSS-5000 & \textbf{0.0654} & \textbf{0.0053} & LSS-2000 & \textbf{0.1130} &  \underline{0.0151} \\ 
\hline
 & \multicolumn{5}{c|}{\textit{Distribution system}} \\
\hline
 & \multicolumn{2}{c|}{\textit{IEEE 69-Bus system}} & -- & \multicolumn{2}{c|}{\textit{CIGRE-LV}~\cite{barsali2014benchmark}} \\
\hline
Linear~\cite{wang2016analytical} & 0.0793 & 0.0058 & Linear & 0.0509 &	0.0028\\
PLinear~\cite{gao2023analytical} & 0.3911 & 0.1335 & PLinear &  0.1079	& 0.0113\\
MC-100 & 0.0707 & 0.0041 & MC-100 & 0.0678 & 0.0037\\
MC-500 & 0.0669 &	0.0039 & MC-1000 & 0.0453 & 0.0022\\
MC-1000  & 0.0657 &	0.0039 & MC-5000 &  0.0438 & 0.0021\\
\underline{LSS-100} & 0.0654 & 0.0038 & \underline{LSS-100} & 0.0507 & 0.0025\\
LSS-500 & 0.0627& 0.0036 & LSS-1000 &  0.0418	& 0.0021\\
LSS-1000 & \textbf{0.0601} &\textbf{0.0033}  & LSS-5000 & \textbf{0.0411} & \textbf{0.0021} \\
\hline
\end{tabular}

\begin{tablenotes}[flushleft]
\footnotesize
\item Example: MC-500/LSS-50 denotes IMNF with MC-based/LSS-based sampling using 500 scenarios.
\end{tablenotes}
\end{threeparttable}
\vspace{-0.5cm}
\end{table}

\subsection{Experimental Results of PPF simulation}

First, we examine how the number of sampled scenarios impacts the accuracy of density estimation. Fig.~\ref{fig:increasingscenarios} illustrates how increasing the number of sampled scenarios improves the density estimation performance of the proposed framework under LSS-based approaches. In addition, Table~\ref{table:expppfcomparision} quantitatively shows that the divergence metrics consistently decrease as the number of sampled scenarios increases for both MC-based and LSS-based density estimation.

Regarding transmission systems, different behaviors are observed across the evaluated test cases. Fig.~\ref{fig:four_panels}~(a) and (b) present the density estimation results for exemplar buses. In the IEEE-39 bus system, the proposed framework demonstrates a clear performance advantage over both benchmark approaches, achieving absolute decreases of 0.2328 in JSD and 0.0911 in TVD relative to the Linear benchmark, and absolute decreases of 0.0459 in JSD and 0.0095 in TVD relative to PLinear. For the CIGRE-HV system, the proposed framework shows a noticeable improvement over PLinear, achieving absolute decreases of 0.0161 in JSD and 0.0043 in TVD. Compared with the Linear benchmark, the performance remains comparable, the proposed method yields a decrease of 0.0018 in JSD and a slight increase of 0.0006 in TVD. The underlying reason for these observations lies in the different degrees of nonlinearity in the two systems. In the IEEE-39 bus system, clear nonlinear characteristics are observed in the PF mapping. This is evidenced by the superior performance of the piecewise linear benchmark compared with the purely linear model, indicating that a single global linear approximation is insufficient. Owing to its greater flexibility in modeling complex nonlinear distributions, the proposed framework further outperforms the piecewise linear benchmark and achieves the best overall performance. In contrast, in the CIGRE-HV system, the PF relationship can be well approximated by a linear function of the form $\mathbf{y} = A\mathbf{x} + \mathbf{b}$. In such cases, more expressive models, such as piecewise linear or highly flexible nonlinear methods, are not necessarily advantageous. Although these models offer increased representational capacity, they may introduce higher data requirements, potential overfitting risks, or parameter redundancy when the underlying system behavior is close to linear. These observations suggest that while the proposed framework maintains performance comparable to Linear models in near-linear regimes, its advantages become more pronounced when the underlying system exhibits stronger nonlinear characteristics.

Another conclusion drawn from the comparison between the IEEE-39 bus system and the CIGRE-HV system is that the required number of scenarios depends on both the complexity of the underlying distribution and the scale of the system. CIGRE-HV system is smaller, and the PF is close to linear. The satisfactory performance can be achieved with a relatively small number of scenarios. As shown in the Table~\ref{table:expppfcomparision}, using only 50 scenarios can already yield a performance better than the Plinear benchmark and close to the Linear benchmark. As shown in Table~\ref{table:expppfcomparision}, using only 50 scenarios (LSS-50) already yields better performance than the PLinear benchmark, with absolute decreases of 0.0089 in JSD and 0.0010 in TVD. Meanwhile, the performance remains close to that of the Linear benchmark, with only marginal gaps of 0.0054 in JSD and 0.0039 in TVD.

Regarding the CIGRE-LV and IEEE-69 bus systems, the proposed framework consistently outperforms all benchmark methods, as summarized in Table~\ref{table:expppfcomparision}. Fig.~\ref{fig:four_panels} (c) and (d) illustrate the density estimation results for representative buses. For instance, in the IEEE-69 bus system, compared with the best LSS-based configuration (LSS-1000), the proposed framework achieves absolute decreases of 0.0192 in JSD and 0.0025 in TVD relative to the Linear benchmark, and 0.3310 in JSD and 0.1302 in TVD relative to PLinear. Consistent with previous observations, we found that even under a limited number of scenarios ($N=100$), the proposed framework remains competitive and surpasses conventional baselines. For example, in the CIGRE-LV system with $N=100$, the proposed framework achieves absolute decreases of 0.0002 in JSD and 0.0003 in TVD compared with Linear, and 0.0572 in JSD and 0.0088 in TVD compared with PLinear.

\vspace{-0.2cm}
\section{Conclusion}

In this paper, we propose a novel framework for voltage-density approximation in PPF. We first demonstrate that the proposed IMNF model (embedded within the overall framework) achieves state-of-the-art performance as a standalone PF solver compared with existing benchmark methods. We then show that the complete framework attains state-of-the-art performance in voltage-density approximation for PPF. Furthermore, we investigate how the number of sampled scenarios affects the final approximation accuracy and demonstrate that LSS-based sampling consistently outperforms MC-based approaches in terms of sample efficiency. Despite these promising results, the application of INNs to PPF still warrants further investigation. In particular, we observe that commonly used INN architectures, primarily designed for generative modeling via maximum likelihood, are not fully optimized for PPF tasks. A key challenge lies in the asymmetry of learning difficulty between the forward and inverse mappings, which motivates the use of bidirectional training. We therefore anticipate that future research will further explore both the theoretical foundations and practical design of INN-based frameworks for PF and PPF applications.

\bibliographystyle{IEEEtran}
\bibliography{ref}
\end{document}